\documentclass{article}
\usepackage{graphicx}

\usepackage{lipsum}
\usepackage[utf8]{inputenc}
\usepackage[T1]{fontenc}
\usepackage{amsmath}
\usepackage{amsthm}
\usepackage{amssymb}
\usepackage{geometry}
\usepackage{mathtools}
\usepackage{xcolor}
\usepackage{bbm}
\usepackage[style=alphabetic, maxbibnames=99]{biblatex}
\usepackage{IEEEtrantools}
\usepackage{enumerate}

\DeclareMathOperator{\wt}{wt}
\usepackage{booktabs}
\usepackage{array} 
\usepackage[below]{placeins} 
\usepackage{caption} 
\usepackage{subcaption}
\usepackage{moreverb}
\usepackage{hyperref}
\usepackage{psfrag}
\usepackage{graphics}
\usepackage{lineno} 
\usepackage{float}
\usepackage{hyperref}
\usepackage[all]{xy} 
\newcommand{\C}{\mathcal{C}}
\newcommand{\F}{\mathbb{F}_2}
\newcommand{\Wtilde}{\widetilde{W}}
\newcommand{\Ytilde}{\widetilde{Y}}
\newcommand{\mc}{\succeq_{\mathrm{mc}}}
\newcommand{\lnoisy}{\succeq_{\mathrm{ln}}}
\newcommand{\rel}{\succeq_{\mathrm{rel}}}

\newcommand{\etamc}{\eta_{\mathrm{mc}}}
\newcommand{\etaln}{\eta_{\mathrm{ln}}}
\DeclareMathOperator{\BEC}{BEC}
\DeclareMathOperator{\BSC}{BSC}
\DeclareMathOperator{\Ber}{Ber}
\DeclareMathOperator{\MAP}{MAP}
\DeclareMathOperator{\Geom}{Geom}
\DeclareMathOperator{\EE}{\mathbb{E}}
\DeclareMathOperator*{\argmax}{arg\,max}
\DeclareMathOperator{\cave}{cave}
\DeclareMathOperator{\con}{con}
\usepackage{cleveref}

\newtheorem{definition}{Definition}[section]
\newtheorem{thm}[definition]{Theorem}
\newtheorem{clm}[definition]{Claim}
\newtheorem{lem}[definition]{Lemma}

\newtheorem{pro}[definition]{Proposition}

\newtheorem{defn}[definition]{Definition}

\usepackage{fullpage}
\date{}
\title{A Quantitative Version of More Capable Channel Comparison}
\author{Donald Kougang-Yombi \\ \texttt{donald.yombi@aims.ac.rw}\\ AIMS Rwanda
\and Jan H{\k{a}}z{\l}a\\ \texttt{jan.hazla@aims.ac.rw}\\ AIMS Rwanda}

\begin{document}
\pagenumbering{arabic}
\maketitle

\begin{abstract}
This paper introduces a quantitative generalization of 
the ``more capable'' comparison of broadcast channels, 
which is termed ``more capable with advantage''. 
Some basic properties are demonstrated (including tensorization
on product channels), and a
characterisation is given for the cases of Binary Symmetric Channel (BSC)
and Binary Erasure Channel (BEC).

It is then applied to two problems. First, a list decoding bound
on the BSC is given that applies to transitive codes that achieve capacity
on the BEC. Second, new lower bounds on entropy rates
of binary hidden Markov processes are derived.
\end{abstract}
\section{Introduction}

In this paper we build on a well-known notion of ``more capable'' ordering
of broadcast channels. This relation was first introduced by K{\"o}rner and
Marton~\cite{KM77} and
has been useful in aspects of channel coding,
see, e.g., Chapter~6 in~\cite{csiszar2011information} and
Chapter~5 in~\cite{NetInf}.

A channel $W:\mathcal{X}\to\mathcal{Y}$ 
is said to be more capable than $\Wtilde:\mathcal{X}\to \widetilde{\mathcal{Y}}$ if
mutual information inequality
$I(X;Y)\ge I(X;\Ytilde)$ holds for all random variables 
$X,Y,\Ytilde$ such that $P_{Y|X}=W$ and $P_{\Ytilde|X}=\Wtilde$.
We will show that this definition gracefully generalizes to 
a notion that we term \emph{more capable with advantage}:
We say that $W$ is more capable than $\Wtilde$ with advantage $\eta$
(which we write
$W+\eta\mc\Wtilde$)
if for every $X,Y,\Ytilde$ as before it holds
$I(X;Y)+\eta\ge I(X;\Ytilde)$. In particular, for every pair of channels
there exists the smallest $\etamc=\etamc(W,\Wtilde)\ge 0$ such that 
$W+\etamc\mc\Wtilde$.

In \Cref{sec:definitions} we show that the more capable with advantage property tensorizes, in particular
if $W+\eta\mc\Wtilde$, then for $N$-fold product channels it holds
$W^N+\eta N\mc\Wtilde^N$. This allows us to obtain a couple of
applications that focus on the Binary Symmetric and Binary Erasure Channels
(BSC and BEC, respectively). In particular, it is well known
that $\BSC_p$ is not more capable than $\BEC_q$ for any non-trivial choice
of the crossover probability $p$ and erasure probability $q$. In contrast, our quantitative notion allows
to leverage bounds obtained for the BEC 
(which in general might be easier to work with) into conclusions
about the BSC. 
Accordingly, in \Cref{sec:bec-bsc} we characterize the more capable
with advantage relation between BEC and BSC. 

A related theorem by Samorodnitsky reads
(see~\cite{hązła2022optimal} for a discussion):
\begin{thm}[Corollary 1.9 in \cite{samorodnitsky2016entropy}]
\label{thm:samorodnitsky}
Let $0\le p,q\le 1$ such that $q\le 4p(1-p)$. Then,
$\BSC_p+h(p)-q\mc\BEC_q$.
\end{thm}

Here and elsewhere, $h(p)=-p\log_2(p)-(1-p)\log_2(1-p)$ denotes the binary entropy function.
Our results are motivated by this theorem. While \Cref{thm:samorodnitsky}
concerns a natural case where the advantage $\etamc=h(p)-q$ is equal to the difference
of the channel capacities, we look for further applications in a more general
situation.

\subsection{List decoding}
We are motivated by the following basic question. Let $\mathcal{C}=\{\mathcal{C}_n\}$ be a
family of binary codes with good performance 
(e.g., vanishing block-MAP error probability)
on the erasure channel $\BEC_q$ for some $0<q<1$. Let $p$ satisfy
$h(p)=q$, so that the channel capacities satisfy
$C(\BSC_p)=C(\BEC_q)$. What can be said about the performance
of $\mathcal{C}$ on the symmetric channel $\BSC_p$?
Our inequalities give a direct bound on the conditional entropy
$H(X|Y_{\BSC})$
of a random codeword $X$ given its noisy version $Y_{\BSC}$ after
transmitting through $\BSC_p$. We then elevate this bound into
an exponential size list-decoding bound under additional assumption that
$\mathcal{C}_n$ is linear and transitive.
We believe that our bounds are the first of this kind and
significantly better than what is implied by existing literature.

\begin{thm}
\label{thm4.1}
    Let $0<p<\frac{1}{2} , 0<q<1$ and $\eta\ge 0$ such that
    $\BSC_p+\eta\mc\BEC_q$.
    
    Let
    $\mathcal{C}=\{\mathcal{C}_n\}$     
    be a family of linear binary codes with increasing block lengths
    and transitive permutation groups.
    Let $X\in\mathbb{F}_2^n$ be
    a random variable uniform over $\mathcal{C}_n$ and $Y_{\BEC}$ and
    $Y_{\BSC}$ the output of transmitting $X$ over $\BEC_q$
    and $\BSC_{p}$, respectively. 
    
    If $H(X|Y_{\BEC})=o(n)$,
    then, for every $\varepsilon>0$, there exists a list decoder
    $D:\mathbb{F}_2^n\to \mathcal{C}_n^L$ with $L\le 2^{(\eta+\varepsilon) n}$  such that:
    \begin{align*}
    \lim_{n\to\infty} \sup_{x\in \mathcal{C}_n} \Pr[X\notin D(Y_{\BSC})\mid X=x]=0\;.
    \end{align*}
\end{thm}
The assumption $H(X|Y_{\BEC})=o(n)$ follows from any list-decoding scheme on the BEC
with vanishing error probability and list size $2^{o(n)}$. In particular,
it is implied by vanishing error probability for bit- or block-MAP
unique decoding on the BEC.

In~\cite{hązła2022optimal} a similar theorem was proved in the cases
of $p$ and $q$ covered by \Cref{thm:samorodnitsky}. However, that does
not include the equal-capacity case $q=h(p)$ that we emphasize here.
Technically, our new result
requires significantly more than substituting the new inequality.
The initial bound on error probability is weaker than in~\cite{hązła2022optimal}
and we need to employ a careful definition of an appropriate list decoder
and sharp threshold properties of Boolean functions.

A graphical illustration of some of the bounds is provided in \Cref{fig:list-decoding}. It can be checked that in the case $q=h(p)$ the coefficient
$\etamc$ does not exceed 0.04.
\begin{figure}[!ht]
    \centering
    \begin{subfigure}[t]{0.45\textwidth}
        \centering
        \includegraphics[width=\linewidth]{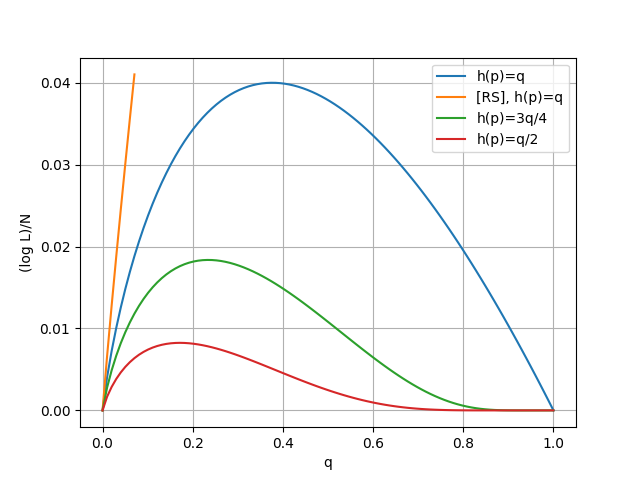}
        \caption{The list decoding exponent for transitive
        codes achieving capacity on the BEC as a function of the
        erasure probability $q$. In the case $q=h(p)$ also
        the bound from~\cite{RS22} is included for illustration.}
        \label{fig:equal-capacity}
    \end{subfigure}
    \begin{subfigure}[t]{0.45\textwidth}
        \centering
        \includegraphics[width=\linewidth]{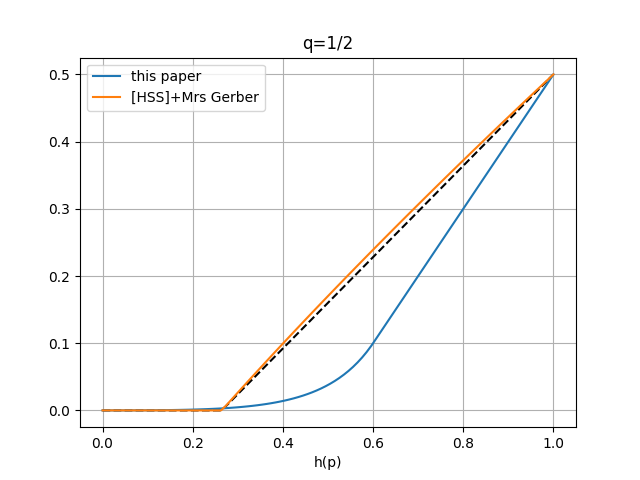}
        \caption{The list decoding exponent for a capacity-achieving
transitive code on the BEC$_{1/2}$ as a function of the BSC
entropy $h(p)$. This (blue) graph becomes a straight
line around $h(p)\approx 0.6$.
The bound implied by~\cite{hkazla2021codes}
is included for illustration. The non-constant part
of the \cite{hkazla2021codes} bound might look like a straight line, but
in fact it is slightly concave.}
    \label{fig:q-half}
    \end{subfigure}
\caption{Illustration of our list decoding bounds.}
\label{fig:list-decoding}
\end{figure}

Let us compare briefly to existing results.  For that purpose
let us assume that the code achieves capacity on the BEC,
that is $R\approx 1-q$, where $R$ is the rate of the code.
Let $\gamma:=\limsup_{n\to\infty}\frac{\log_2 L}{n}$ be the
constant in the list size exponent under MAP decoding. 
\Cref{thm4.1} essentially
says that $\gamma\le\etamc(\BSC_p,\BEC_q)$. It is readily seen
that it is much better than the trivial bound
$\gamma\le \min(R, h(p))$. 

Rao and Sprumont \cite{RS22}
showed that 
$\gamma\le\max\left(p\log_2\frac{2}{R},4p\right)$ for any
transitive linear code. \Cref{fig:equal-capacity} compares this
bound to ours. (Of course the comparison is not fair since their bound applies
to \emph{all} transitive codes, regardless of their performance on the
BEC.)

Furthermore, \cite{hkazla2021codes} shows that linear codes
with (sufficiently fast) vanishing error probability on $\BEC_q$ also
have vanishing error probability on $\BSC_{p'}$ for
$p'<p_0:=1/2-\sqrt{2^{q-1}(1-2^{q-1})}$. For $p>p_0$, one can use
that fact to obtain a lower bound on $H(Y_{\BSC_{p'}})$,
then apply Mrs Gerber's lemma to lower bound $H(Y_{\BSC_p})$,
and finally follow our method in \Cref{sec:list-decoding} to get a list decoding
bound. However, around $q=h(p)$ this seems to give a much worse bound 
(see \Cref{fig:q-half} for $q=0.5$).

\subsection{Hidden Markov process entropy rate}
\label{sec:intro-markov}

Second, we provide new lower bounds on the entropy rate of binary hidden Markov
processes. Let $0<q,\alpha<1/2$. Let $W_2, W_3, \ldots, Z_1, Z_2, \ldots, X_1$
be independent random variables with distributions
$W_i\sim\Ber(q), Z_i\sim\Ber(\alpha), X_1\sim\Ber(1/2)$. Then,
let $X_{n+1}:=X_{n}\oplus W_{n+1}$ and $Y_n:=X_n\oplus Z_n$,
where $\oplus$ denotes addition in $\mathbb{F}_2$.
The problem is to estimate the \emph{entropy rate}
\begin{align*}
    \overline{H}(Y):=\lim_{n\to\infty} \frac{H(Y_1,\ldots,Y_n)}{n}
\end{align*}
In~\cite{ordentlich2016novel}, Ordentlich used a version of Samorodnitsky's result to 
prove a lower bound:
\begin{thm}[Theorem~2 in \cite{ordentlich2016novel}]
\label{thm:ord}
Let $\gamma=4\alpha(1-\alpha)$, $G\sim\Geom(1-\gamma)$ and
denote by $q^{*k}:=\frac{1-(1-2q)^k}{2}$ the $k$-fold convolution of $q$.
Then, it holds
\begin{align*}
\overline{H}(Y)\ge
h\big(\alpha * h^{-1}(\EE h(q^{*G}))\big)\;,
\end{align*}
where $h(p)$ is the binary entropy function
and $\alpha*\beta=\alpha(1-\beta)+(1-\alpha)\beta$ is the convolution operation.
\end{thm}
As with list decoding, we can replace the inequality by Samorodnitsky
with ``more capable with advantage''
inequalities for the whole range
of erasure channels. This results in
strictly better lower bounds on $\overline{H}(Y)$,
which, however, still do not reach the true rate.
An illustration is provided in \Cref{fig1}.

\begin{thm}
\label{pro5.3}
    In the setting above, it holds:

\begin{align}
\label{eq1.4}
    \overline{H}(Y)\ge \underset{0\le \gamma\le 1}{\sup}\big((1-\gamma)\mathbb{E} h(q^{*G})+h(\alpha)-\etamc(\BSC_\alpha,\BEC_\gamma)\big),
\end{align}
where 
$G\sim \mathrm{Geom}(1-\gamma)$.
\end{thm}
\begin{figure}[H]
\centering
\begin{subfigure}[t]{.45\textwidth}
  \centering
  \includegraphics[width=\linewidth]{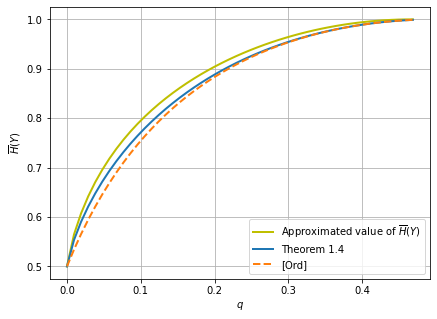}
  \caption{$\alpha=0.11$}
  \label{fig:sub1}
\end{subfigure}
\begin{subfigure}[t]{0.45\textwidth}
  \centering
  \includegraphics[width=\linewidth]{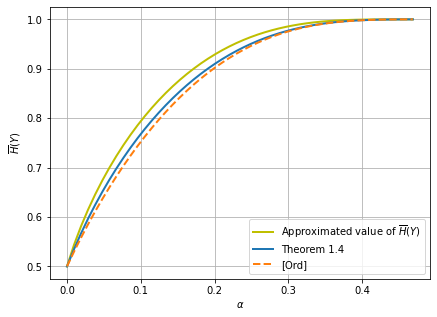}
  \caption{$q=0.11$}
  \label{fig:sub2}
\end{subfigure}
\caption{Comparison of our lower bounds for $\overline{H}(Y)$, that given in \cite{ordentlich2016novel}, and
the true rate $\overline{H}(Y)$ as computed
by an approximation algorithm.}
\label{fig1}
\end{figure}
\paragraph{Organization}
In the rest of the paper we proceed as follows.
In \Cref{sec:definitions,sec:bec-bsc} we give definitions
and basic properties,
and characterize the ``more capable with advantage'' relation between the BSC and BEC.
We also discuss an analogous ``less noisy with advantage'' relation
and show that it also tensorizes and has a clean characterization in terms
of the mutual information difference function.
Then we proceed to applications: list decoding in \Cref{sec:list-decoding}
and hidden Markov processes in \Cref{sec:hidden-markov}.

\section{Definitions and basic properties}
\label{sec:definitions}
Here and below,
assume that we are given
two discrete memoryless channels (DMC) $W$ and $\Wtilde$ over input alphabet $\mathcal{X}$
and random variables
$X,Y,Z$ such that $P_{Y|X}=W$ and $P_{Z|X}=\Wtilde$.

\begin{defn}
\label{def2.1}
    Channel $W$ is said to be more capable than $\Wtilde$ 
    (written as $W\mc\Wtilde$) if for every $X,Y,Z$ as above it holds
    $I(X;Y)\ge I(X;Z)$.

    Let $\eta\ge 0$. We will say that $W$ with advantage $\eta$
    is more capable than $\Wtilde$ and write
    $W+\eta\mc\Wtilde$ if for every $X,Y,Z$ it holds
    $I(X;Y)+\eta\ge I(X;Z)$.
\end{defn}

The original paper~\cite{KM77} also introduced the ``less noisy'' ordering
of channels, which also found many uses, see, e.g., discussion
in~\cite{makur2018comparison, polyanskiy2016strong}. We also make
an analogous definition in the less noisy case:

\begin{defn}
    We say that $W$ is less noisy than $\Wtilde$ and
    write $W\lnoisy\Wtilde$ if for every Markov chain
    $U\to X\to(Y,Z)$ it holds
    $I(U;Y)\ge I(U;Z)$.
    For $\eta\ge 0$, we say that $W$ with advantage $\eta$
    is less noisy than $\Wtilde$ and write
    $W+\eta\lnoisy\Wtilde$ if for every Markov chain $U\to X\to(Y,Z)$ it holds
    $I(U;Y)+\eta\ge I(U;Z)$.
    \label{def3}
\end{defn}

From the definition, being more capable (less noisy) with advantage $\eta$ implies being more capable (less noisy) with any larger advantage $\eta'\ge\eta$. Therefore, it makes sense to introduce:
\begin{definition}
Let $\mathrm{rel}\in\{\mathrm{mc},\mathrm{ln}\}$.
Let $\eta_{\mathrm{rel}}(W,\Wtilde)$ be the infimum of all $\eta$ such
that $W+\eta\rel \Wtilde$. 
\end{definition}
By continuity considerations, it holds that
$W+\etamc(W,\Wtilde)\mc\Wtilde$. Also note that
$0\le\etamc(W,\Wtilde)\le C(\Wtilde)$, where $C(\Wtilde)$
denotes the Shannon capacity of $\Wtilde$. Furthermore, for discrete
memoryless symmetric channels (which maximize mutual input/output information
for uniform distribution), it holds 
$\etamc\ge C(\Wtilde)-C(W)$.

It turns out that both properties generalize
in a natural way to the case of product channels.
The proofs are straightforward generalizations of the case 
without advantage. We state the result here and provide
proofs in the appendix.

\begin{pro}
\label{mc-tensorize}
Let $\mathrm{rel}\in\{\mathrm{mc},\mathrm{ln}\}$.
Let $W_1+\eta_1\rel\Wtilde_1$ and $W_2+\eta_2\rel\Wtilde_2$. Then,
$W_1\times W_2+(\eta_1+\eta_2)\rel\Wtilde_1\times\Wtilde_2$.
In particular,
for every $n\ge 1$,
if $W+\eta\rel\Wtilde$, then $W^n+\eta n\rel\Wtilde^{n}$.
\end{pro}

In analyzing our relations between two channels we will use the following
natural function.
Let $\Delta(\mathcal{X})$ be the set of probability distributions
on $\mathcal{X}$. As is common, we will think of it as a convex
subset of a linear space:
\begin{definition}
\label{def f}
We let $f=f(W,\Wtilde):\Delta(\mathcal{X})\to\mathbb{R}$ as
\begin{align*}
    f(P) := I(X;Y) - I(X;Z)\;,
\end{align*}
where $X,Y,Z$ are random variables as above such that $P_X=P$.
\end{definition}

The values $\etamc$ and $\etaln$ can be naturally
characterized in terms of function $f$. First, the following claim follows
immediately from the definitions:

\begin{clm}
\label{cl:mc-f-char}
$\eta_{mc}(W,\Wtilde)=-\inf_P f(P)$.
\end{clm}

In the less noisy case the characterisation is more interesting.
\begin{defn}
    Let $I$ be a convex set and 
    $g: I\to\mathbb{R}$ be a function, 
    then the (upper) concave envelope of $g$ is defined as:
    \begin{align*}
        g^{\cave}(x)&=\inf\{h(x), h\ \text{concave},\ h(y)\ge g(y)\, \forall y\in I\}
    \\&=\sup \big\{\sum_{i=1}^{k}\lambda_i g(x_i); 0\le\lambda_1,\cdots,\lambda_k\le1,  \sum_{i=1}^{k}\lambda_i=1,\sum_{i=1}^{k}\lambda_i x_i=x\big\}.
    \end{align*}
\end{defn}
It is well-known that $g^{\cave}$ is a concave function and for every $r\in I, g^{\cave}(r)\ge g(r)$ \cite{nair2013upper}.
\begin{pro}
\label{lem:ln-f-char}
    $\etaln(W,\Wtilde)=\sup_P f^{\cave}(P)-f(P)$. 
\end{pro}
\begin{proof}
    Let $\eta=\sup_P f^{\cave}(P)-f(P)$, we want to show that 
    $\eta=\eta_{\mathrm{ln}}$. We will show it in two steps.
    First, we show $W+\eta\lnoisy\Wtilde$. Then we show that
    $W$ with advantage $\eta-\varepsilon$ is not less noisy
    than $\Wtilde$ for any $\varepsilon>0$.

    For the first step, we want to show
    $I(U;Z)-I(U;Y)\le\eta$ for every Markov chain
    $U\to X\to (Y,Z)$. Indeed,
    \begin{IEEEeqnarray*}{rCll}
    I(U;Z)-I(U;Y)&=&I(X;Z)-I(X;Z|U)-\big(I(X;Y)-I(X;Y|U)\big)
    \qquad&\text{by Markov property,}
    \\
    &=&\EE_U\left[ f\left(P_{X|U}\right)\right] - f(P_X)&\\
    &\le&\EE_U\left[ f^{\cave}\left(P_{X|U}\right)\right] - f(P_X)
    &\text{since $f^{\cave}\ge f$}\\
    &\le& f^{\cave}(P_X)-f(P_X)\le\eta
    & \text{since $f^{\cave}$ is concave.}
    \end{IEEEeqnarray*}

    On the other hand, let $P$ be a distribution on $\mathcal{X}$
    that realizes $f^{\cave}(P)-f(P)=\eta$. By the sup-definition of
    concave envelope there exist some $0\le\lambda_1,\ldots,\lambda_k$,
    $\sum_{i=1}^k\lambda_i=1$ and $P_1,\ldots,P_k$, $\sum_{i=1}^k\lambda_i P_i=P$ such that
    $\sum_{i=1}^k\lambda_if(P_i)>f(P)+\eta-\varepsilon$.
    Now we can specify the joint distribution $(U,X)$ by $\Pr[U=i]=\lambda_i$
    and $P_{X|U=i}=P_i$ and compute as before
    \begin{align*}
        I(U;Z)-I(U;Y)=\sum_{i=1}^k\lambda_i f(P_i)-f(P)>\eta-\varepsilon
    \end{align*}
    Hence, $W$ with advantage $\eta-\varepsilon$ is \emph{not} less noisy than
    $\Wtilde$, as claimed.
\end{proof}

\section{Characterization for the $\BEC_q$ and the $\BSC_p$}
\label{sec:bec-bsc}
In this section, we give a characterization of the notions of more capable (respectively less noisy)
with advantage in the classical case of
Binary Symmetric Channel (BSC) and the Binary Erasure Channel (BEC). 
In the following, $p$ will denote the crossover probability of the BSC and $q$ the erasure probability of the BEC.

We did not find closed-form expressions for
$\eta_{\mathrm{mc}}$ and $\eta_{\mathrm{ln}}$ in all cases,
but we can characterize them using \Cref{cl:mc-f-char} and \Cref{lem:ln-f-char} and
the properties of the mutual information difference function $f$.
Accordingly, for $r\in[0,1]$ and $X_r\sim \Ber(r)$, let $Y_{\BEC}$ and $Y_{\BSC}$ be the outputs of transmitting $X_r$ through the $\BEC_q$ and $\BSC_p$, respectively. 
In this section, given $q$ and $p$, we let
\begin{align*}
f(r):=I(X_r;Y_{\BSC})-I(X_r;Y_{\BEC})=h(r* p)-h(p)-(1-q)h(r)\;.
\end{align*}

In the following we will use a lemma summarizing the properties of
$f$. Most of this analysis can be found e.g., in Claim~3 in~\cite{nair2009capacity},
but we need a slightly stronger statement.
We provide the full proof in the appendix. Various shapes of $f$ and its
concave/convex envelopes for $q=1/2$ can be seen in \Cref{f}. 
\begin{figure}[H]
        \centering
        \includegraphics[width=0.9\linewidth]{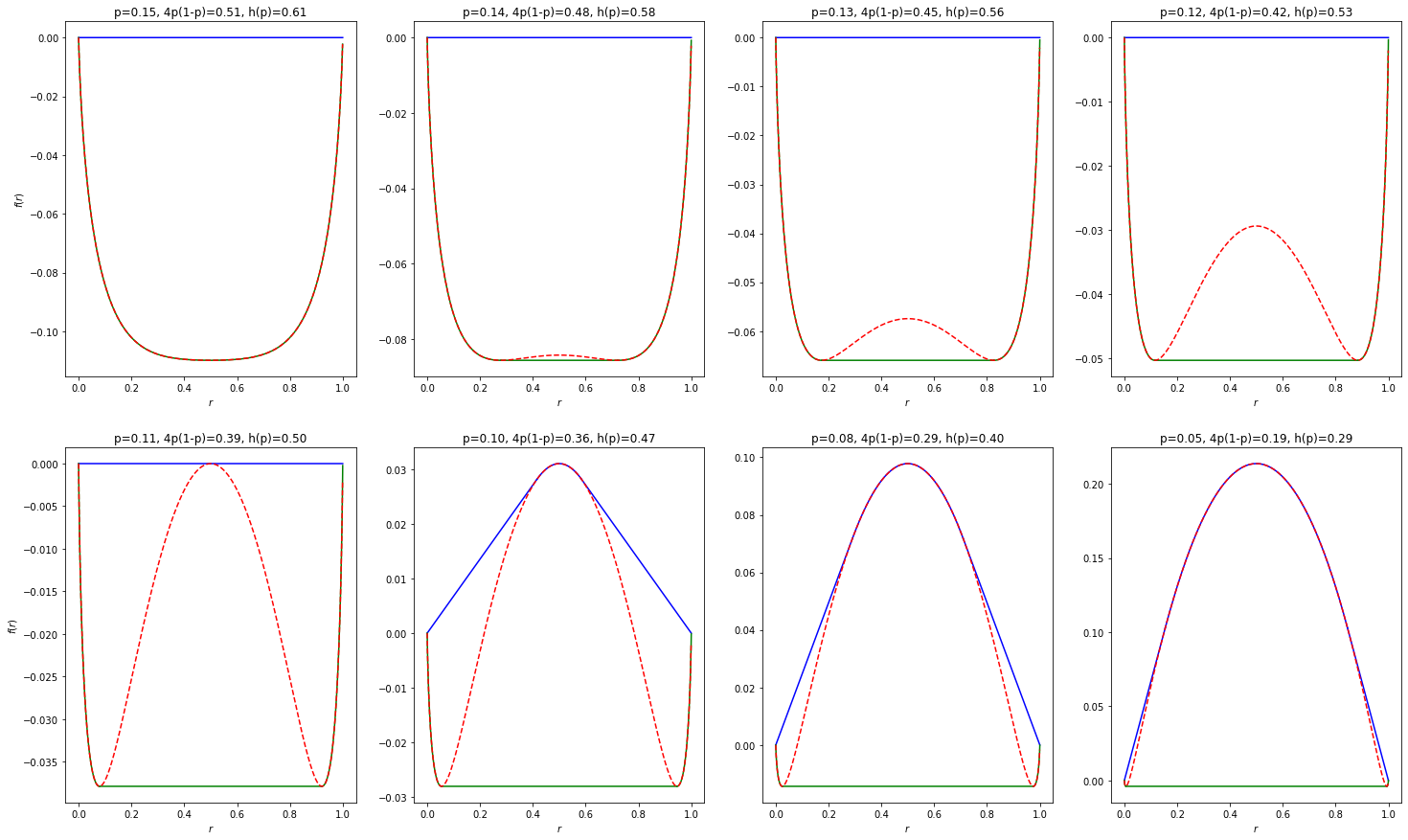}
        \caption{Shapes of $f$ and its concave (in blue) and convex (in green) envelopes for $q=0.5$ and different values of $p$.}
        \label{f}
\end{figure}

\begin{lem}
\label{pro1appx}
   For $p,q\in(0,1)$, on the interval
   $0\le r\le 1/2$, we have:\footnote{
   For $r>1/2$ observe that $f(r)=f(1-r)$.}
    \begin{itemize}
        \item For $ h(p)<q$,
        there exists $r_0$ in $(0,\frac{1}{2})$ such that $f$ decreases from $0$ to $r_0$ and then increases from $r_0$ to $1/2$. Furthermore, there exists
        $r'\in(r_0,1/2)$ such that $f$ is convex on $[0,r']$ and concave on
        $[r',1/2]$. The  maximum of $f$ is at $r=1/2$.        
        \item For $4p(1-p) < q \le h(p)$, $f$ has the same properties as
        in the preceding point,
        but now the maximum is at $r=0$.
        \item For $q\le 4p(1-p)$, $f$ is decreasing from $0$ to $1/2$. 
        Furthermore, $f$ is convex.
    \end{itemize}
\end{lem}

We can now state characterizations of the two relations.

\subsection{More capable with advantage}
\begin{pro}
Let $p,q\in(0,1)$. The following give the $\etamc$
values between $\BEC$ and $\BSC$:
\begin{enumerate}[1)]
    \item $\BEC_q+\max\big(0,q-h(p)\big)\mc \BSC_p$.
    \item If $q\le 4p(1-p)$, then $\BSC_p+h(p)-q\mc \BEC_q$.
    \item If $4p(1-p)< q$, then $\BSC_p+[-f(r_0)]\mc \BEC_q$.
\end{enumerate}
\end{pro}
\begin{proof}
\begin{enumerate}[1)] 
    \item
    By~\Cref{cl:mc-f-char}, we have
    $\etamc(\BEC_q,\BSC_p)=\max_{0\le r\le 1}f(r)$.
    And from~\Cref{pro1appx}, the maximum of $f$ is either at $r=0$
    or $r=1/2$. So,
    $\etamc(\BEC_q,\BSC_p)=\max\big(f(0),f(1/2)\big)=
    \max\big(0, q-h(p)\big)$.
    \item 
    First, $\etamc(\BSC_p,\BEC_q)=-\min_{r}f(r)$
    by~\Cref{cl:mc-f-char}. If $q\le 4p(1-p)$, then
    $-\min_r f(r)=-f(1/2)=h(p)-q$.
    \item If $q>4p(1-p)$, then $-\min_r f(r)=-f(r_0)$. \qedhere
\end{enumerate}
\end{proof}

\subsection{Less noisy with advantage}
\begin{pro}
Let $p,q\in (0,1)$. The values of $\etaln$ are given by:
\begin{enumerate}[1)]
    \item $\BEC_q\lnoisy \BSC_p$ if $q\le 4p(1-p)$.
    \item $\BEC_q+q-h(p)-f(r_0)\lnoisy\BSC_p$ if $4p(1-p)<q$.
    \item If $q\le 4p(1-p)$, then $\BSC_p+h(p)-q\lnoisy\BEC_q$.
    \item If $4p(1-p)<q\le h(p)$, then
    $\BSC_p-f(r_0)\lnoisy\BEC_q$.
    \item If $h(p)<q$, then there exists $\eta> -f(r_0)$ such
    that $\BSC_p+\eta\lnoisy\BEC_q$.
\end{enumerate}
\end{pro}
\begin{proof}
Again, Figure \ref{f} can be consulted for illustration.
    \begin{enumerate}[1)]
        \item By \Cref{lem:ln-f-char}, 
        $\eta_{\mathrm{ln}}(\BEC_q,\BSC_p)=\max_{0\le r\le 1} 
        f(r)-f^{\con}(r)$, where $f^{\con}$ is the lower convex envelope
        of $f$ (see \Cref{f}). This is because of the known
        fact that $(-f)^{\cave}=-f^{\con}$. By \Cref{pro1appx},
        if $q\le 4p(1-p)$, then
        $f$ is convex and therefore $f=f^{\con}$ and $\eta_{\mathrm{ln}}=0$.
        \item If $4p(1-p)<q$, then from \Cref{pro1appx}:
        \begin{align*}
        f^{\con}(r)=\begin{cases}
        f(r)&\text{if $r\le r_0$ or $r\ge 1-r_0$,}\\
        f(r_0)&\text{otherwise.}
        \end{cases}
        \end{align*}
        Therefore, $\eta_{\mathrm{ln}}=\max_r f^{\con}(r)-f(r)=f(1/2)-f(r_0)
        =q-h(p)-f(r_0)$.
        \item By \Cref{lem:ln-f-char},
        $\eta_{\mathrm{ln}}(\BSC_p,\BEC_q)=\max_r f^{\cave}(r)-f(r)$.
        If $q\le 4p(1-p)$, then $f$ is convex, so its concave envelope is $f^{\cave}\equiv 0$ (the straight line from $f(0)$ to $f(1)$).       
        So by \Cref{pro1appx}, 
        $\eta_{\mathrm{ln}}=-\min_r f(r)=-f(1/2)=h(p)-q$.
        \item For $4p(1-p)< q\le h(p)$ we still have $f^{\cave}\equiv 0$, but $\underset{0\le r\le 1}{\max}\ -f(r)=-f(r_0)$, so
        $\eta_{\mathrm{ln}}=-\min_r f(r)=-f(r_0)$.        
        \item For $h(p)<q$, from \Cref{pro1appx} it can be seen that
        the concave envelope has the following form. There exists a unique
        point $r_1\in [r_0,1/2]$ such that $\frac{f(r_1)}{r_1}=f'(r_1)$. 
        Then:
        \begin{align*}
        f^{\cave}(r)=\begin{cases}
        \frac{f(r_1)}{r_1}r&\text{if $r\le r_1$,}\\
        f(r)&\text{if $r_1<r<1-r_1$,}\\
        \frac{f(r_1)}{r_1}(1-r)&,\text{if $r\ge 1-r_1$}.
        \end{cases}
        \end{align*}
        In other words, the envelope $f^{\cave}$ consists of a line from 
        $(0,0)$ to $(r_1,f(r_1))$ and then is equal to $f$ until $r=1/2$
        (and symmetrically for $r>1/2$).

        Since $f(1/2)>f(0)$, we can check that $f^{\cave}(r_1)>0$ and        
        $f^{\cave}(r_0)>0$. Consequently,
        \begin{align*}
        &\eta_{\mathrm{ln}}=\max_r f^{\cave}(r)-f(r)\ge f^{\cave}(r_0)-f(r_0)>-f(r_0)\;.
        \qedhere
        \end{align*}
    \end{enumerate}
\end{proof}

\section{List decoding}
\label{sec:list-decoding}
Recall that a binary code $\mathcal{C}$ with block length $n$
is called \emph{linear} if it is a linear subspace
of $\mathbb{F}_2^n$ and \emph{transitive} if its symmetry
group $\mathcal{G}(\mathcal{C})=\{\pi\in S_n:\pi\circ \mathcal{C} = \mathcal{C}\}$ is transitive, i.e., for every $1\le i,j\le n$, there exists $\pi\in\mathcal{G}(\mathcal{C})$ such that $\pi(i)=j$.
For $x,y\in\mathbb{F}_2^n$, we write
$\Delta(x,y)=\{1\le i\le n:x_i\ne y_i\}$ for the Hamming distance
between $x$ and $y$. We also write $x\le y$ if $x_i\le y_i$
for every coordinate. 

We call a function $f:\mathbb{F}_2^n\to\mathbb{R}$
increasing (resp.~decreasing) if $x\le y$ implies $f(x)\le f(y)$
(resp.~$f(x)\ge f(y)$).
We call $f$ transitive symmetric if its symmetry group
$\mathcal{G}(f)=\{\pi\in S_n:\pi\circ f=f\}$ is transitive.

This section is  devoted to proving Theorem \ref{thm4.1}.
The proof proceeds in two parts.
First, we provide an upper bound on the list-decoding probability of error using a similar strategy as in~\cite{hązła2022optimal}:
\begin{thm}
\label{thm4.2}
   Let $0<p\le\frac{1}{2}$, $0\le q\le 1$ and $\eta$ such that $\BSC_p+\eta\mc \BEC_q$. Let $\{\mathcal{C}_n\}$ be a family of binary codes. 

   Let $X\in\mathbb{F}^n_2$ be a random variable uniform over $\mathcal{C}_n$, $Y_{\BEC}$ and $Y_{\BSC}$ the output of transmitting $X$ over $\BEC_q$ and $\BSC_p$ respectively. Assume that $H(X|Y_{\BEC})=o(n)$.
   
   Then, for all $\delta>0$, there exists a list decoder $D:\mathbb{F}_2^n\to \mathcal{C}_n^L$ such that $$\Pr[X\notin D(Y_{\BSC})]\le P_e+o(1)$$ where $P_e=\frac{1-h(p)+\eta}{1-h(p)+\eta+\delta}<1$ and the size of the list is $L\le 2^{(\eta +\delta)n}$.
\end{thm}
Theorem \ref{thm4.2} shows the existence of a list-decoder with the required list size, but the error probability  is not vanishing,
but rather just bounded away from 1. To fix that, we will use a classical threshold result appearing in many works including \cite{kudekar2016reedmuller}. This will allow us to prove:

\begin{thm}
\label{Thm1}
   Let $\{\mathcal{C}_n\}$ be a family of transitive binary linear codes, and let $\delta,\varepsilon>0$.
   Let $D:\mathbb{F}_2^{n}\to\mathcal{C}_n^{L}$ be a $L$-list decoder for $\mathcal{C}_n$ for some $L=L(n)$. Let $X\in\F^n$ be a random uniform codeword in $\mathcal{C}_n$ and $Y_p$ the output of transmitting $X$ on the $\BSC_p$. 
   
   If $\Pr[X\notin D(Y_p)]\le 1-\varepsilon+o(1)$, then there exists a $L'$-list decoder $D':\mathbb{F}_2^{n}\to\mathcal{C}_n^{L'}$ (with $L'\le \frac{4}{\varepsilon}$L) such that,
   on the $BSC_{p-\delta}$,
   \begin{align*}
   \lim_{n\to\infty}\sup_{x\in\mathcal{C}_n}\Pr[X\notin D'(Y_{p-\delta})\;\vert\;X=x]=0.
   \end{align*}
\end{thm}

We will now prove \Cref{thm4.1} assuming
\Cref{thm4.2} and \Cref{Thm1}. Then, we proceed 
with the proofs of those two theorems. We first need
a claim about continuity of $\etamc$, which is proved in the appendix.

\begin{clm}
\label{clm4.3}
    Let $0\le q\le 1$, then the function $p\longmapsto \etamc(\BSC_p,\BEC_q)$ is continuous on $]0,\frac{1}{2}]$.
\end{clm}
\begin{proof}[Proof of Theorem \ref{thm4.1}]
    Let $0<p<\frac{1}{2}, 0< q< 1$ and $\eta$ such that $\BSC_p+\eta\mc \BEC_q$.
    Let $\mathcal{C}=\{\C_n\}$ be a family of transitive linear binary codes with increasing block lengths. Let $X\in\mathbb{F}_2^n$ be a random variable uniform over $\C_n$ and $Y_{\BEC}$ and $Y_{\BSC}$ the output of transmitting $X$ over $\BEC_q$ and $\BSC_{p}$, respectively. Let's also assume that $H(X|Y_{\BEC})=o(n)$. 

    Let $\varepsilon>0$. Let $\etamc(p,q)=\etamc(\BSC_p,\BEC_q)$.
    By \Cref{clm4.3}, there exists $\delta>0$ such that $\etamc(p+\delta,q)\le \etamc(p,q)+\frac{\varepsilon}{3}$.
    Let us apply Theorem \ref{thm4.2} for $\BSC_{p+\delta}$,
    $\etamc(p+\delta,q)$ and $\varepsilon/3$. From there, there exist
    $0\le P_e<1$ and a list-decoder $D_0: 
     \mathbb{F}_2^n\to\mathcal{C}_n^{L_0}$, with 
     \begin{align*}
     L_0\le 2^{\big(\etamc(p+\delta,q)+\frac{\varepsilon}{3}\big)n}
     \le 2^{\left(\etamc(p,q)+\frac{2\varepsilon}{3}\right)n}
     \le 2^{\left(\eta+\frac{2\varepsilon}{3}\right)n}
     \end{align*}
     such that \begin{equation*}
        \Pr[X\notin D_0(Y_{\BSC_{p+\delta}})]\le P_e+o(1).
    \end{equation*}
Let's write $P_e=1-\varepsilon'$, such that $\Pr[X\notin D_0(Y_{\BSC_{p+\delta_0}})]\le 1-\varepsilon'+o(1)$. Then by Theorem \ref{Thm1}, there exists a list-decoder $D:\mathbb{F}_2^n\to\mathcal{C}_n^L$ with $L\le \frac{4}{\varepsilon'}L_0$, 
\begin{equation}
\label{eq4}
    \lim_{n\to \infty}\sup_{x\in\C_n}\Pr[X\notin D(Y_{\BSC_{p}})|X=x]=0.
\end{equation}
Since for $n$ large enough $L\le\frac{4}{\varepsilon'}L_0\le 2^{(\eta+\varepsilon)n}$
(and for smaller $n$ we can use any other decoder without affecting the limit
in~\eqref{eq4}), the conclusion follows.
\end{proof}

\subsection{Proof of \Cref{thm4.2}}
Let's denote $Y:=Y_{\BSC}=X\oplus Z$ with $Z$ which is i.i.d. $\Ber(p)$ and independent of $X$. Furthermore, let $R_n$, be the rate of $\mathcal{C}_n$.
We have 
\begin{align*}
    H(Y)&=H(X,Y)-H(X|Y)
    \\&=H(X,Z)-H(X|Y)
    \\&=H(X)+H(Z)-H(X|Y),&\ \text{because $X$ and $Z$ are independent,}
    \\&\ge R_nn +h(p)n-\eta n-H(X|Y_{\BEC}),&\ \text{since $\BSC_p+\eta\mc \BEC_q$,}
    \\&=R_nn +h(p)n-\eta n- o(n),&\ \text{since $H(X|Y_{\BEC})=o(n)$,}
    \\&=(R_n +h(p)-\eta) n- o(n),
\end{align*}
so 
\begin{align}
\label{eq1}
    H(Y)\ge(R_n +h(p)-\eta) n- o(n).
\end{align}
Let $\mathcal{L} := \{(x, z) : x \in \mathcal{C}_n, z\in\mathbb{F}_2^n, \wt(z) < p n + n^{\frac{3}{4}}\}$, and for $y\in \mathbb{F}_2^n$, let $\mathcal{B}_y:=\{(x, z) : x \in \mathcal{C}_n, z\in\mathbb{F}_2^n: x\oplus z=y\}$. Let $\delta>0$ fixed, we will say that $y\in \mathbb{F}_2^n$ is $\delta$-likely if the size of the set $\mathcal{B}_y\cap\mathcal{L}$ is greater than $2^{(\eta+\delta)n}$. Let $P_n(\delta)$ be the probability that the random string $Y$ is $\delta$-likely. If $(x,z)\in\mathcal{L}$, we have 
\begin{align*}
    \Pr[X=x,Z=z]&=\Pr[X=x]\Pr[Z=z]
    \\&\ge 2^{-R_n n}p^{p n+n^{\frac{3}{4}}}(1-p)^{(1-p)n-n^{\frac{3}{4}}}
    \\&=2^{-R_n n} 2^{-h(p)n}\big(\frac{p}{1-p}\big)^{n^{\frac{3}{4}}}.
\end{align*}
If $y\in\mathbb{F}_2^n$ is $\delta$-likely then
\begin{align*}
    \Pr[Y=y]&\ge\sum\limits_{(x,z)\in \mathcal{B}_y\cap\mathcal{L}}\Pr[X=x,Z=z]
    \\&\ge 2^{n(\eta+\delta)}\times 2^{-n(R_n+h(p))} \big(\frac{p}{1-p}\big)^{n^{\frac{3}{4}}}
    \\&=2^{-n(R_n+h(p)-\eta-\delta)} \big(\frac{p}{1-p}\big)^{n^{\frac{3}{4}}},
\end{align*}
so we get an upper bound on the entropy of $Y$ as follows
\begin{align*}
    H(Y)&\le P_n(\delta)\big(n(R_n+h(p)-\eta-\delta)+n^{\frac{3}{4}}\log(\frac{1-p}{p})\big)+\sum\limits_{y\in\mathbb{F}_2^n,\ y\ \text{not}\ \delta\text{-likely} }\Pr[Y=y]\log\frac{1}{\Pr[Y=y]}
    \\&\le P_n(\delta)\big(n(R_n+h(p)-\eta-\delta)+n^{\frac{3}{4}}\log(\frac{1-p}{p})\big)+ n(1-P_n(\delta))+1,\ \text{same as in the proof of Lemma 18 in \cite{hązła2022optimal}},
    \\&\le n-nP_n(\delta)\big(1-R_n-h(p)+\eta+\delta\big)+ P_n(\delta)n^{\frac{3}{4}}\log(\frac{1-p}{p})+1,
    \\&=n-nP_n(\delta)\big(1-R_n-h(p)+\eta+\delta\big) +o(n).
\end{align*}
So we have 
\begin{align}
   H(Y)&\le n-nP_n(\delta)\big(1-R_n-h(p)+\eta+\delta\big) +o(n),
   \nonumber
   \\&\le n+nR_n-nP_n(\delta)(1-h(p)+\eta+\delta)+o(n),
    \label{equ2}
\end{align}
then using the Inequalities \eqref{eq1} and \eqref{equ2} we get
\begin{align*}
    (R_n +h(p)-\eta) n- o(n)\le H(Y)\le n+nR_n-nP_n(\delta)\big(1-h(p)+\eta+\delta\big) +o(n).
\end{align*}
Then, it holds asymptotically
\begin{align*}
    P_n(\delta)\le \frac{1-h(p)+\eta}{1-h(p)+\eta+\delta}+o(1).
\end{align*}
Now consider a list decoder $D$ such that, for each $y\in\mathbb{F}_2^n$ which is not $\delta$-likely, it outputs a list of all codewords $x\in\mathcal{C}_n$ such that $(x,x+y)\in\mathcal{L}\cap\mathcal{B}_y$. So if $Y$ is not $\delta$-likely and $(X,Z)\in\mathcal{L}$, then the decoder $D$ is successful and $X\in D(Y)$. 
Therefore we have 
\begin{align*}
    \Pr[X\notin D(Y)]&=\Pr[Y\ \text{is $\delta$-likely or $(X,Z)\notin\mathcal{L}$} ]
    \\&\le\Pr[Y\ \text{is $\delta$-likely}]+\Pr[(X,Z)\notin\mathcal{L}]
    \\&\le \frac{1-h(p)+\eta}{1-h(p)+\eta+\delta}+o(1)+\Pr[\wt(Z)\ge p n + n^{\frac{3}{4}} ]
    \\ &= \frac{1-h(p)+\eta}{1-h(p)+\eta+\delta}+ o(1),
\end{align*}
since $\Pr[\wt(Z)\ge p n + n^{\frac{3}{4}} ]=o(1)$ by the Hoeffding's inequality.
Since by construction it is easy to see that $D$ outputs at
most $2^{(\eta+\delta)n}$ codewords, it is the desired list decoder.
\qed

\subsection{Proof of \Cref{Thm1}}
The proof of \Cref{Thm1} requires some preliminary results that we will give first. In the following
let $\mathcal{C}$ be a binary linear code.
It is well known (and not hard to see) that the error probability on the BSC for $L$-list decoding is minimized by the following randomized MAP decoder:
\begin{defn}[MAP decoder for $L$-list decoding]
\label{def1}
    Let $y\in\mathbb{F}_2^n$, and $N_d=|\{x\in\C , \Delta(x,y)=d\}|$, let $d^{*}=d^*(y)$ be the smallest $d$ such that $\sum\limits_{i=0}^{d}N_i\ge L$. Let's denote $$S_1(y)=\{x\in\C, \Delta(x,y)<d^{*}\},$$ $$S_2(y)=\{x\in\C, \Delta(x,y)=d^{*}\},$$ $$W(y)=L-|S_1(y)|.$$ 
    The randomized list $\MAP$-decoding of $y$ is given by $D_{\MAP}(y)=S_1(y)\cup T(y)$, where $T(y)$ is obtained by choosing $W(y)$ elements in $S_2(y)$ uniformly at random.
\end{defn}
Accordingly, in the proof of Theorem~\ref{Thm1} we will assume without loss of generality that
$D$ is the decoder from Definition~\ref{def1}.
\begin{clm}
\label{clm1}
    $\forall x\in\C,\forall z\in\F^n, 0\in S_1(z)$ 
    if and only if $x \in S_1(x+z)$ (and similarly for $S_2$).
\end{clm}
\begin{proof}
    Let $x\in\C$ and $z\in\F^n$,
    \begin{align*}
        0\notin S_1(z) &\iff |\{c\in\C, \Delta(c,z)\le\Delta(0,z)\}|\ge L,\ \text{by Definition \ref{def1}}
        \\&\iff |\{c\in\C, \Delta(c+x,z+x)\le\Delta(x,z+x)\}|\ge L
        \\&\iff x\notin S_1(x+z),\ \text{because $c+x$ is a codeword since $\C$ is linear.}
    \end{align*}
The proof for $S_2(z)$ is similar.
\end{proof}
\begin{clm}
\label{clm4.5}
    $\forall x\in\C, \forall z\in\mathbb{F}_2^n$, $$\Pr[0\in D_{\MAP}(z)]=\Pr[x\in D_{\MAP}(x+z)].$$
\end{clm}
\begin{proof}
   This follows by Claim \ref{clm1}.
\end{proof}
\begin{lem}
\label{pro3}
    For $y\in\mathbb{F}_2^n$, let's define  $P(y):=\Pr[0\in D_{\MAP}(y)]$. The function $P(y)$ is a decreasing function on $\mathbb{F}_2^n$.
\end{lem}
\begin{proof}
By definition of $D_{\MAP}$, we have
    $$P(y)=\begin{cases}
        1&\text{if}\ 0\in S_1(y),
        \\ \frac{W(y)}{|S_2(y)|}&\text{if}\ 0\in S_2(y),
        \\0 & \text{otherwise.}
    \end{cases}$$
    Let $y_1, y_2\in\F^n$ such that $y_1\le y_2$, we want to prove that $P(y_1)\ge P(y_2)$. Let's consider for a given $y$ the numbers $$A_{y}=|\{x\in\C, \Delta(x,y)<\Delta(0,y)\}|\ \text{and}\ B_{y}=|\{x\in\C, \Delta(x,y)\le\Delta(0,y)\}|.$$ 
    We now state a simple claim:
    \begin{clm}
\label{lem2}
     Let $y_1, y_2\in\F^n$ such that $y_1\le y_2$. If there exists $c\in \C$ such that $\Delta(c,y_1)\le \Delta(0,y_1)$ then we also have $\Delta(c,y_2)\le \Delta(0,y_2)$.
\end{clm}
    By the definition of $S_1$ and $S_2$ (Definition \ref{def1}) and 
    \Cref{lem2}, the following statements are true: 1) $0\in S_1(y)$ if and only if $B_y<L$. 2) 
    $0\in S_1(y)\cup S_2(y)$ if and only if  $A_y<L$. 3) $A_{y_1}\le A_{y_2}$ and $B_{y_1}\le B_{y_2}$. 4) If $0\in S_2(y)$ then $A_y=|S_1(y)|$ and $B_y=|S_1(y)|+|S_2(y)|$.

For the proof of the proposition we proceed through all possible cases: 
\begin{itemize}
    \item If $0\in S_1(y_1)$ or $0\notin S_1(y_2)\cup S_2(y_2)$, we obviously have $P(y_1)\ge P(y_2)$.
    \item If $0\notin S_1(y_1)$, then by 1) and 3), it follows that $0\notin S_1(y_2)$.
    \item If $0\in S_2(y_1)$ and $0\in S_2(y_2)$, then by 4) and the definition of $P(y)$ we have
    \begin{align*}
        P(y_1)&=\frac{L-A_{y_1}}{B_{y_1}-A_{y_1}}
        \\&\ge \frac{L-A_{y_1}}{B_{y_2}-A_{y_1}},\ \text{since}\  B_{y_1}\le B_{y_2}
        \\&\ge \frac{L-A_{y_2}}{B_{y_2}-A_{y_2}},\ \text{since}\ A_{y_1}\le A_{y_2}
        \\&=P(y_2)\;.
    \end{align*}
    \item If $0\notin S_1(y_1)\cup S_2(y_1)$, then $A_{y_1}\ge L$, so $A_{y_2}\ge L$ thus $0\notin S_1(y_2)\cup S_2(y_2)$ therefore $P(y_1)=P(y_2)=0$.
\end{itemize}
Therefore in all cases if $y_1\le y_2$, then $P(y_1)\ge P(y_2)$.
\end{proof}
To apply the KKL theorem, we need a deterministic decoder for $\C$,
which we will now define.
\begin{defn}
\label{defD'}
    Let $0\le \alpha\le 1$, let $y\in\F^n$, $S_1(y)$ and $S_2(y)$ as in Definition \ref{def1}. Let's consider the decoder defined as 
    \begin{equation*}
        D_\alpha(y)=\begin{cases}
            S_1(y)& \text{if}\ \frac{W(y)}{|S_2(y)|}<\alpha
            \\S_1(y)\cup S_2(y)& \text{if}\ \frac{W(y)}{|S_2(y)|}\ge \alpha
        \end{cases}.
    \end{equation*}
\end{defn}
\begin{clm}
    \label{list-size}
    Let $0< \alpha\le 1$. If $L'=\max_{y\in\mathbb{F}_2^n} |D_\alpha(y)|$ then $L'\le\frac{1}{\alpha}L$.
\end{clm}
\begin{proof}
Let $y_0=\argmax_{y\in\mathbb{F}_2^n} |D_\alpha(y)|$,  if $D_\alpha(y_0)=S_1(y_0)$, then $|D_\alpha(y_0)|\le L$. In the other case, we have $W(y_0)\ge \alpha|S_2(y_0)|$, so $L=W(y_0)+|S_1(y_0)|\ge \alpha|S_2(y_0)|+|S_1(y_0)|$, then $L\ge \alpha(|S_2(y_0)|+|S_1(y_0)|)=\alpha L'$, and therefore $L'\le \frac{1}{\alpha}L$.  
\end{proof}
\begin{clm}
\label{prop18}
For every $x\in \C, \Pr[x\in D_\alpha(x+Z)]=\Pr[0\in D_\alpha(Z)],$ for $Z\underset{iid}{\sim} Ber(p)$.
\end{clm}
\begin{proof}
   By Claim \ref{clm1} we have, $\forall x\in \C, \forall z\in \F^n, 0\in D_p(z)\text{ if and only if }x\in D_p(x+z)$. The claim follows
   by averaging over $z$.
\end{proof}
\begin{lem}
\label{prop19}
Let $Z\underset{iid}{\sim} Ber(p)$ for some $0\le p\le 1$
and $\varepsilon>0$. If $\Pr[0\in D_{\mathrm{\MAP}}(Z)]\ge\varepsilon$,
then $\Pr[0\in D_{\varepsilon/2}(Z)]\ge\frac{\varepsilon}{2}$.
\end{lem}
\begin{proof}
    By the definition of the MAP decoder:
   \begin{equation*}
       \varepsilon\le\Pr[0\in D_{\MAP}(Z)]=\Pr[0\in S_1(Z)]+\mathbbm{E}\left[\frac{W(Z)}{|S_2(Z)|}\mathbbm{1}[0\in S_2(Z)]\right]\;.
   \end{equation*} 

On the other hand for $Z\underset{iid}{\sim} \Ber(p)$ and $\alpha=\varepsilon/2$ we have 
\begin{align*}
    \Pr[0\in D_\alpha(Z)]=&\Pr[0\in S_1(Z)]+\Pr\left[0\in S_2(Z)\ \text{and}\ \frac{W(Z)}{|S_2(Z)|}\ge \alpha\right]
    \\\ge&
    \Pr[0\in S_1(Z)]+ \mathbbm{E}\left[\frac{W(Z)}{|S_2(Z)|}\mathbbm{1}[0\in S_2(Z)]\mathbbm{1} \left[\frac{W(Z)}{|S_2(Z)|}\ge\alpha\right]\right]
    \\=&  
    \Pr[0\in D_{\MAP}(Z)]-\mathbbm{E}\left[\frac{W(Z)}{|S_2(Z)|}\mathbbm{1}[0\in S_2(Z)]\mathbbm{1} \left[\frac{W(Z)}{|S_2(Z)|}<\alpha\right]\right]
    \\\ge&\varepsilon-\alpha=\varepsilon/2\;.\qedhere
\end{align*}
\end{proof}
From now on 
let's fix some $\alpha>0$, denote $D'=D_{\alpha}$ 
and consider the function defined on $\F^n$ by $f(y)=\mathbbm{1}[0\in D'(y)]$.
\begin{lem}
    The function $f$ is decreasing on $\F^n$.
    \label{pro20}
\end{lem}
\begin{proof}
    Let $y_1, y_2\in\F^n$ such that $y_1\le y_2$. Note that $f(y)=1$ if and only if
    $P(y)\ge\alpha$. But, by \Cref{pro3} it holds $P(y_1)\ge P(y_2)$.
    The claim follows.
\end{proof}
\begin{lem}
    If code $\mathcal{C}$ is transitive, then 
   function $f$ is likewise transitive symmetric.
   \label{pro21}
\end{lem}
\begin{proof}
    Let $i,j\in[n], i\neq j$, since $\C$ is transitive there exists $\pi\in \mathcal{S}_n$ such that $\pi(\C)=\C$ and $\pi(i)=j$. We want to show that $f(\pi(y))=f(y)$ for every $y$ in $\F^n$. Let $y\in\F^n$, first let's prove that $S_1(\pi(y))=\pi(S_1(y))$ and $S_2(\pi(y))=\pi(S_2(y))$. Let $x\in \C$,
    \begin{align*}
        x\in S_1(\pi(y)) &\iff |\{c\in\C, \Delta(c,\pi(y))\le\Delta(x,\pi(y))\}|<L
        \\&\iff |\{c\in\C, \Delta(\pi^{-1}(c),y)\le\Delta(\pi^{-1}(x),y)\}|<L
        \\&\iff \pi^{-1}(x)\in S_1(y)
        \\&\iff x\in \pi(S_1(y)).
    \end{align*}
    Then we have $S_1(\pi(y))=\pi(S_1(y))$. Similarly we can show that $S_2(\pi(y))=\pi(S_2(y))$. These two facts show that $D'(\pi(y))=\pi(D'(y))$, therefore we have
    \begin{align*}
        f(\pi(y))&=\mathbbm{1}[0\in D'(\pi(y))]
        \\&=\mathbbm{1}[0\in \pi(D'(y))]
        \\&=\mathbbm{1}[\pi^{-1}(0)\in D'(y)]
        \\&=\mathbbm{1}[0\in D'(y)]
        \\&=f(y).\qedhere
    \end{align*}
\end{proof}

We now restate a classical theorem on sharp thresholds which is
basically an application of the KKL theorem \cite{kahn1989influence}.
\begin{thm}[Theorem 19 in \cite{kudekar2016reedmuller}]
\label{thm19}
There exists a universal constant $C>0$ such that for any
nonconstant
increasing transitive-symmetric Boolean function $g:\mathbb{F}_2^n\to\{0,1\}$
and $0<\varepsilon\le\frac{1}{2}$
\begin{equation*}
    p_{1-\varepsilon}-p_{\varepsilon}\le \frac{C\log(\frac{1-\varepsilon}{\varepsilon})}{\log M},
\end{equation*}
where $p_t=\inf\{p\in[0,1], \Pr[g(Z)=1]\ge t\text{ with }Z \text{ i.i.d.~from }  \Ber(p)\}$.
\end{thm}

\begin{proof}[Proof of Theorem \ref{Thm1}]
Let's suppose $\Pr[X\notin D(Y_p)]\le 1-\varepsilon+o(1)$, so we have $\liminf_{n\to\infty} \Pr[X\in D(Y_p)]\ge\varepsilon$. Then for $n$ large enough $\Pr[X\in D(Y_p)]\ge\frac{\varepsilon}{2}$, then $\Pr[X\in D_{\MAP}(Y_p)]\ge \frac{\varepsilon}{2}$, since $D_{\MAP}$ is the maximum likelihood decoder. By Claim \ref{clm4.5} we have $\Pr[0\in D_{\MAP}(Z_p)]=\Pr[X\in D_{\MAP}(Y_p)]\ge\frac{\varepsilon}{2}$. Let $D':=D_{\frac{\varepsilon}{4}}:\mathbb{F}_2^n\to\C_n^{L'}$, by Claim \ref{list-size} $L'\le\frac{4}{\varepsilon}L$, so applying Lemma \ref{prop19}, we get 
\begin{equation}
\label{eq3}
    \Pr[0\in D'(Z_p)]\ge\frac{\varepsilon}{4}.
\end{equation}
Let $g(y)=\mathbbm{1}[0\notin D'(y)]$, by 
\Cref{pro20} and \Cref{pro21}, function $g$ is increasing and transitive-symmetric. If $g(1^n)=0$, then $\Pr[0\notin D'(Z_p)]=0$.
Otherwise, $g$ is not constant and we let
$p_t$ to be the unique $p\in[0,1]$ such that $\mathbb{E}[g(Z_{p})]=t$. 

Let $0< t<\frac{\varepsilon}{4}$, by Equation \ref{eq3} we have $p\le p_{1-\frac{\varepsilon}{4}}\le p_{1-t}$. Applying \Cref{thm19} to $g$, we get $p_{t}\ge p_{1-t}-\frac{C\log(\frac{1-t}{t})}{\log n }\ge p-\frac{C\log(\frac{1-t}{t})}{\log n }$, and for $n$ large enough $p_{t}\ge p-\delta$, then $\mathbb{E}[g(Z_{p-\delta})]\le t$. Since this holds
for every $t>0$, it follows
\begin{equation*}
    \lim_{n\to\infty} \Pr[0\notin D'(Z_{p-\delta})]=
    \lim_{n\to\infty} \mathbb{E}[g(Z_{p-\delta})]=0.
\end{equation*}
Also, by Claim \ref{prop18} $\Pr[0\notin D'(Z_{p-\delta})]=\sup_{x\in\C_n} \Pr[X\notin D'(Y_{p-\delta})|X=x]$, therefore 
\begin{equation*}
    \lim_{n\to\infty} \sup_{x\in\C_n} \Pr[X\notin D'(Y_{p-\delta})|X=x]=0.
    \qedhere
\end{equation*}
\end{proof}

\section{Hidden Markov processes}
In this section we prove and discuss a new lower bound on the entropy rate of binary hidden Markov processes. In \cite{ordentlich2016novel}, a lower bound has been provided using Samorodnitsky's inequality, here we use the more capable with advantage relation to obtain better bounds as illustrated in Figure \ref{fig1}.
\Cref{pro5.3} follows easily from the following two lemmas:

\label{sec:hidden-markov}
\begin{lem}
\label{pro5.1}
    Let $X\in\mathbb{F}_2^n$ be a random variable. Let $Y,E$ be the results of transmitting $X$ over
$\mathrm{BSC}_\alpha$ and $\mathrm{BEC}_\gamma$, respectively. Let $\eta(\alpha,\gamma):=\etamc(\BSC_\alpha,\BEC_\gamma)$. Then, $$\frac{H(Y)}{n}\ge\frac{I(X;E)}{n}+h(\alpha)-\eta(\alpha,\gamma).$$
\end{lem} 
\begin{proof}
    This follows by Definition \ref{def2.1}, Proposition \ref{mc-tensorize} and the chain rule of the mutual information and entropy. Indeed, we have $\frac{H(Y)}{n} = \frac{I(X;Y)+H(Y|X)}{n}\ge \frac{I(X;E)-\eta(\alpha,\gamma)n}{n}+h(\alpha)=\frac{I(X;E)}{n}+h(\alpha)-\eta(\alpha,\gamma)$.
\end{proof}
Now considering the following Markov process: $X_1$ is uniform in $\mathbb{F}_2$, $X_{t+1}=X_t\oplus W_{t+1}$, where $W_t\sim\mathrm{Ber}(q)$ are iid. $Y_t=X_t\oplus Z_t$, where $Z_t\sim\mathrm{Ber}(\alpha)$ are also iid. In~\cite{ordentlich2016novel} it was shown:

\begin{lem}[Proposition 2 in \cite{ordentlich2016novel}]
\label{pro5.2}
    Let $0<\gamma<1$.
    For $X\xrightarrow{\mathrm{BEC}_\gamma}E$ and $G\sim\mathrm{Geom}(1-\gamma)$ (that is, $\Pr[G=g]=\gamma^{g-1}(1-\gamma)$ for $g=1,2,3,\ldots$):
$$ \lim_{n\to\infty} \frac{I(X;E)}{(1-\gamma)n} = \mathbb{E} h(q^{*G})$$
where $q^{*k}=q*q\ldots*q$ for $k$ times, in particular $q^{*k}=\frac{1-(1-2q)^k}{2}$.
\end{lem}
\Cref{pro5.3} follows quickly from
\Cref{pro5.1} and \Cref{pro5.2}.

A reader might have noticed that the bound in \Cref{thm:ord} is
\emph{not} identical to the bound from \Cref{pro5.3}
evaluated at $\gamma=1-(1-2\alpha)^2$.
This is due to the fact that for this special value of
$\gamma$ a stronger ``Mrs Gerber style'' bound is available
in~\cite{samorodnitsky2016entropy} (compare Theorem~1.6
and Corollary~1.9 therein). We did not obtain such
``Mrs Gerber'' bound for general $\gamma$. However, it appears
that the flexibility in the choice of $\gamma$ in
\Cref{pro5.3} overcomes this disadvantage. While we do not have a proof,
the cases that we empirically checked are consistent with
a hypothesis that \Cref{pro5.3} gives stronger bounds
than \Cref{thm:ord}.

In particular, in \Cref{fig3} we plot
the values of $\gamma$ which achieve the maximum
on the right-hand side in \Cref{pro5.3}.
In the two cases studied it appears that the optimal
value of $\gamma$ is generally larger than
$\gamma=4\alpha(1-\alpha)$ employed in
\Cref{thm:ord}.

\begin{figure}[H]
\centering
\begin{subfigure}[t]{.45\textwidth}
  \centering
  \includegraphics[width=\linewidth]{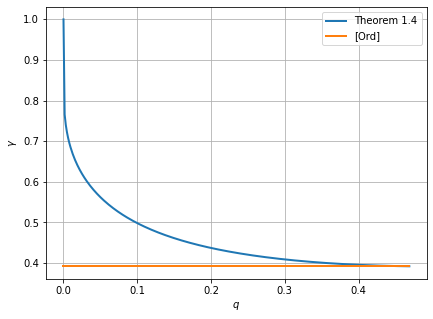}
  \caption{$\alpha=0.11$}
  \label{fig:sub11}
\end{subfigure}
\begin{subfigure}[t]{0.45\textwidth}
  \centering
  \includegraphics[width=\linewidth]{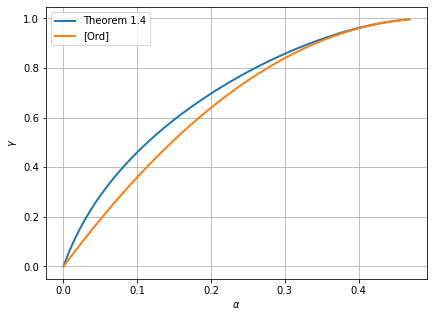}
  \caption{$q=0.11$}
  \label{fig:sub22}
\end{subfigure}
\caption{Comparison of $\gamma$ from \cite{ordentlich2016novel} \big(i.e. $\gamma=4\alpha(1-\alpha)$\big) and that of \Cref{pro5.3} which is the argmax of the right-hand side of \Cref{eq1.4}.}
\label{fig3}
\end{figure}
\appendix

\paragraph{Acknowledgments}
This work was supported by the Alexander von Humboldt Foundation German research chair
funding and the associated DAAD project No. 57610033.

\section{Appendix}

\subsection{Proof of \Cref{mc-tensorize}}

\begin{lem} Let $X^n, Y^n$ and $Z^n$ be three random vectors such that $P_{Y^n|X^n}=W_1\times\ldots\times W_n$ and 
$P_{Z^n|X^n}=\Wtilde_1\times\ldots\times\Wtilde_n$, then
    \begin{equation*}
        I(X^n;Y^n)-I(X^n;Z^n)=\big[I(X_n;Y_n|Y^{n-1})- I(X_n;Z_n|Y^{n-1})\big]+\big[I(X^{n-1};Y^{n-1}|Z_n)- I(X^{n-1};Z^{n-1}|Z_n)\big].
    \end{equation*}
\label{lemappx1}
\end{lem}
\begin{proof}
    We have
    \begin{align*}
        A:&=I(X^n;Y^n)-I(X^n;Z^n)\\&=H(X^n)-H(X^n|Y^n)-H(X^n)+H(X^n|Z^n)
        \\&=H(X^n|Z^n)-H(X^n|Y^n)\\&=H(X_n,X^{n-1}|Z^n)-H(X_n,X^{n-1}|Y^n)
        \\&=H(X^{n-1}|Z^n)+H(X_n|Z^n,X^{n-1})-H(X_n|Y^n)-H(X^{n-1}|Y^n,X_n)
        \\&=\underbrace{H(X^{n-1}|Z^n)-H(X_n|Y^n)}_{C}+H(X_n|Z^n,X^{n-1})-H(X^{n-1}|Y^n,X_n)
    \end{align*}
    \begin{align*}
        B:&=\big[I(X_n;Y_n|Y^{n-1})- I(X_n;Z_n|Y^{n-1})\big]+\big[I(X^{n-1};Y^{n-1}|Z_n)- I(X^{n-1};Z^{n-1}|Z_n)\big]
        \\&=\big[H(X_n|Y^{n-1})-H(X_n|Y^n)-H(X_n|Y^{n-1})+H(X_n|Z_n,Y^{n-1})\big]\\ &\qquad+\big[H(X^{n-1}|Z_{n})-H(X^{n-1}|Z_n,Y^{n-1})-H(X^{n-1}|Z_{n})+H(X^{n-1}|Z^n)\big]
        \\&=\underbrace{H(X^{n-1}|Z^n)-H(X_n|Y^n)}_{C}+H(X_n|Z_n,Y^{n-1})-H(X^{n-1}|Z_n,Y^{n-1}).
    \end{align*}
    Therefore, to show $A=B$ it remains to prove
    \begin{align*}
       H(X_n|Z^n,X^{n-1})+H(X^{n-1}|Z_n,Y^{n-1})=H(X_n|Z_n,Y^{n-1})+H(X^{n-1}|Y^n,X_n).
    \end{align*}
The last equality is always true, indeed we have
\begin{align}
    H(X_n|Z^n,X^{n-1})+H(X^{n-1}|Z_n,Y^{n-1})&= H(X_n|Z_n,X^{n-1},Y^{n-1})+H(X^{n-1}|Z_n,Y^{n-1})\label{a}\\&=H(X^n|Z_n,Y^{n-1}),\nonumber
\end{align}
Equation \eqref{a} holds since $X_{n}$ is conditionally independent of  $Y^{n-1}$ and $Z^{n-1}$ given $X^{n-1}$, similarly
\begin{align}
    H(X_n|Z_n,Y^{n-1})+H(X^{n-1}|Y^n,X_n)&= H(X_n|Z_n,Y^{n-1})+H(X^{n-1}|Y^{n-1},X_n,Z_n)\label{b}
    \\&=H(X^n|Z_n,Y^{n-1}),\nonumber
\end{align}
we have \eqref{b} since $X^{n-1}$ is conditionally independent of  $Y_{n}$ and $Z_{n}$ given $X_{n}$. 
\end{proof}

\paragraph{More capable with advantage}
Let's suppose $W_1+\eta_1\mc \Wtilde_1$ and $W_2+\eta_2\mc\Wtilde_2$. Let $X^2,Y^2$ and $Z^2$ such that 
$W_1\times W_2=P_{Y_1Y_2|X_1X_2}$ and 
$\Wtilde_1\times\Wtilde_2=P_{Z_1Z_2|X_1X_2}$.
By Lemma \ref{lemappx1}, we have 
\begin{equation*}
        I(X^2;Y^2)-I(X^2;Z^2)+(\eta_1+\eta_2) =\big[I(X_2;Y_2|Y_1)- I(X_2;Z_2|Y_1)+\eta_2\big]+\big[I(X_1;Y_1|Z_2)- I(X_1;Z_1|Z_2)+\eta_1\big].
\end{equation*}
Observe that for any conditioning $Y_1=y_1$, the conditional
distributions still obey $P_{Y_2|X_2}=W_2$ and $P_{Z_2|X_2}=\Wtilde_2$.
Therefore, applying $W_2+\eta_2\mc\Wtilde_2$, we get $I(X_2;Y_2|Y_1=y_1)-I(X_2;Z_2|Y_1=y_1)+\eta_2\ge 0$.
Averaging,
$I(X_2;Y_2|Y_1)-I(X_2;Z_2|Y_1)+\eta_2\ge 0$.
Similarly,
$I(X_1;Y_1|Z_2)-I(X_1;Z_1|Z_2)+\eta_1\ge 0$.

All in all,
$I(X^2;Y^2)-I(X^2;Z^2)+(\eta_1+\eta_2)\ge 0$.
Hence, $W_1\times W_2+(\eta_1+\eta_2)\mc\Wtilde_1\times\Wtilde_2$.
\qed

\paragraph{Less noisy with advantage}
    Let's suppose that $W_1+\eta_1\lnoisy \Wtilde_1$ and $W_2+\eta_2\lnoisy\Wtilde_2$. Consider a Markov chain $U\to X^2\to(Y^2,Z^2)$, we have by the chain rule
    \begin{align}
    \label{eq2}
        I(U;Z^2)=I(U; Z_1)+I(U; Z_2|Z_1).
    \end{align}
   For every conditioning $Z_1=z_1$,
   the sequence $U\to X_2\to(Y_2,Z_2)$ is a Markov chain (i.e. $P_{UY_2Z_2|X_2Z_1}=P_{Y_2Z_2|X_2Z_1}\times P_{U|X_2Z_1}$). Since $W_2+\eta_2\lnoisy \Wtilde_2$ and  $U\to X_2\to(Y_2,Z_2)|Z_1$ is a Markov chain, then $I(U; Z_2|Z_1)\le I(U; Y_2|Z_1)+\eta_2$, hence by Equation \ref{eq2} we have 
   \begin{align*}
       I(U;Z^2)\le I(U;Z_1)+ I(U; Y_2|Z_1)+\eta_2
       = I(U;Z_1Y_2)+\eta_2\;.
   \end{align*}Similarly $I(U; Y_2,Z_1)=I(U; Y_2)+I(U;Z_1| Y_2)$ and we can also check that $U\to X_1\to(Y_1,Z_1)|Y_2$ is a Markov chain, then since $W_1+\eta_1\lnoisy \Wtilde_1$ we have $I(U;Z_1| Y_2)\le I(U;Y_1| Y_2)+\eta_1$ and hence $I(U; Z_1Z_2)\le I(U; Y_2) + I(U;Y_1| Y_2)+\eta_1+\eta_2= I(U; Y^2)+\eta_1+\eta_2$.
\qed

\subsection{Proof of \Cref{pro1appx}}
    First, consider the special case $p=1/2$. Then,
    $f(r)=-(1-q)h(r)$. Since $0<q<1$, we are in the case $q<4p(1-p)$,
    and indeed function $f$ is strictly decreasing and convex for
    $r\in]0,1/2[$.

    From now on assume $p\ne 1/2$.
    The function $f$ is differentiable on $]0,\frac{1}{2}[$ and we have 
\begin{align*}
    f'(r)&=(r* p)' h'(r* p)-(1-q)h'(r)\\&=(1-2p)\log\big(\frac{1-r* p}{r* p}\big)-(1-q)\log\big(\frac{1-r}{r}\big),\ \text{recall that}\ h'(r)=\log\big(\frac{1-r}{r}\big).
\end{align*}
Similarly the function $f'$ is differentiable on $]0,\frac{1}{2}[$ and we have 
\begin{align*}
    f''(r)&=(1-2p)^2 h''(r* p)-(1-q)h''(r)\\&=\frac{1}{\ln{2}}\Bigg[\frac{-(1-2p)^2}{(r* p)(1-r* p)}+\frac{1-q}{r(1-r)}\Bigg],\ \text{recall that}\ h''(r)=-\frac{1}{\ln{2}}\frac{1}{r(1-r)}
    \\&=\frac{1}{\ln{2}}\Bigg[\frac{(1-2p)^2qr^2-q(1-2p)^2r+ p(1-q)(1-p)}{(r* p)(1-r* p)r(1-r)}\Bigg].
\end{align*}
Let us consider the function (which is the numerator of $f''$), $N(r)=(1-2p)^2qr^2-q(1-2p)^2r+ p(1-q)(1-p),$ we have 
$N'(r)=(1-2p)^2q(2r-1)< 0$ (for $r\in]0,\frac{1}{2}[$), so $N$ is strictly decreasing, then $f''$ has at most one root in $]0,\frac{1}{2}[$. This implies that $f'$ has at most one extremum in $]0,\frac{1}{2}[$. On the other hand, since $f'(\frac{1}{2})=0$ then there exist at most one $r_0\in]0,\frac{1}{2}[$ such that $f'(r_0)=0$ (otherwise $f'$ would have more than one extremum), therefore $f$ also has at most one extremum in $]0,\frac{1}{2}[$. Using the fact that $\lim_{r\to 0}{f'(r)}=-\infty$ (since $q<1$),
this extremum cannot be a maximum. Therefore, we can conclude that
\begin{align}
    \underset{0\le r\le\frac{1}{2}}{\max}\ f(r)&=\max\big(f(0),f(\frac{1}{2})\big)\\&=\label{max_f}\max\big(0,q-h(p)\big).
\end{align}
This solves the problem of where the maximum is achieved. On the other hand, by computing the discriminant of the quadratic function $N(r)$, we get $\Delta=q(1-2p)^2\big(q-4p(1-p)\big)$.
\begin{itemize}
    \item For $q\le 4p(1-p)$, we have $\Delta\le 0$, which implies
    that $N(r)$ does not have a zero on $r\in ]0,1/2[$, hence
    $f''>0$, so $f$ is convex and $f'$ is strictly increasing on $[0,\frac{1}{2}]$. Since $f'(\frac{1}{2})=0$, we also have $f'\le 0$, so $f$ 
    is strictly decreasing on $[0,\frac{1}{2}]$.
    \item For $q> 4p(1-p)$, then $\Delta>0$, since the quadratic $N(r)$
    achieves minimum at $r=1/2$, that means there exists $r'\in]0,1/2[$ such that on $]0,r'[$, $f''>0$ and on $]r',\frac{1}{2}]$, $f''<0$. So $f$ is convex on $[0,r']$ and concave on $[r',\frac{1}{2}]$. Also $f'$ is strictly increasing on $[0,r']$ and strictly decreasing on $[r',\frac{1}{2}]$. Since $f'(\frac{1}{2})=0$ and $\lim_{r\to 0}{f'(r)}=-\infty$, then $f'(r')>0$, so there exists $r_0\in]0,r'[$ such that $f'(r_0)=0$. Therefore $f$ is decreasing on $[0,r_0]$ and increasing on $]r_0,\frac{1}{2}]$.\qed
\end{itemize}

\subsection{Proof of \Cref{clm4.3}}
    Let $p\in ]0,\frac{1}{2}]$. Similar as in
    \Cref{sec:bec-bsc}, let 
    $f_{p,q}(r)=h(r*p)-h(p)-(1-q)h(r)$. Also let
    $\etamc(p,q)=\etamc(\BSC_p,\BEC_q)$.
    
    Let $\delta_0>0$ be small enough such that
    $0<p-\delta_0<p$ and consider the interval $I=[p-\delta_0,p+\delta_0]\cap[0,1/2]$.
    Let $\delta$ be such that $p+\delta\in I$. The conclusion will follow if we show that 
    $|\etamc(p+\delta,q)-\etamc(p,q)|$ goes to 0 as $|\delta|$ goes to 0. 
    We have 
    \begin{align*}
        |\etamc(p+\delta,q)-\etamc(p,q)|
        &=\left|\sup_{0\le r\le 1}f_{p+\delta,q}(r)-\sup_{0\le r\le 1}f_{p,q}(r)\right|&\\
        &\le \sup_{0\le r\le 1}|f_{p+\delta,q}(r)-f_{p,q}(r)|\\
        &=\sup_{0\le r\le 1}|h\big(r*(p+\delta)\big)-h(r*p)+h(p+\delta)-h(p)|\\
        &\le|\delta|\left(\sup_{x\in I}|h'(x)|+
        \sup_{x\in I,r\in[0,1]}\left|
        \frac{dh(r*x)}{dx}\right|\right)\\&\le C|\delta|\;.&\qed
    \end{align*}

\printbibliography
\end{document}